\documentclass[12pt]{amsart}

\usepackage{graphicx}
\usepackage{amssymb}

\numberwithin{equation}{section}

\usepackage{enumitem}
\usepackage{amssymb, amsmath,cite}    
\usepackage{amsthm}
\usepackage{xcolor}
\usepackage[labelfont={color=black,footnotesize},textfont={color=black,footnotesize,it,up}]{caption} 
\usepackage{url}
\usepackage{bm}
\usepackage[titletoc,toc,title]{appendix}

\theoremstyle{plain}
\newtheorem{prop}{Proposition}[section]
\newtheorem{thm}[prop]{Theorem}

\newtheorem{lemma}[prop]{Lemma}
\newtheorem{con}[prop]{Conjecture}

\theoremstyle{definition}
\newtheorem{Def}{Definition}[section]
\newtheorem{remark}{Remark}[section]

\newcommand{\norm}[1]{\left\lVert#1\right\rVert}
\newcommand{\abs}[1]{\left|#1\right|}

\usepackage{subfig,tikz}
\usetikzlibrary{decorations.pathmorphing}
\usepackage{graphics}

\newcommand{\td}{\text{d}}
\def\be{\begin{equation}}
\def\ee{\end{equation}}
\def\bea{\begin{eqnarray}}
\def\eea{\end{eqnarray}}

\begin{document}

\title{Remarks on mass and angular momenta for $U(1)^2$-invariant initial data}

\author{Aghil Alaee}
\address{Department of Mathematics and Statistics\\ Memorial University of Newfoundland\\St John's NL A1C 4P5, Canada}
\email{aak818@mun.ca }

\author{Hari K. Kunduri}
\address{Department of Mathematics and Statistics\\ Memorial University of Newfoundland\\St John's NL A1C 4P5, Canada}
\email{hkkunduri@mun.ca}

\date{\today}

\begin{abstract}
We extend Brill's positive mass theorem to a large class of asymptotically flat, maximal, $U(1)^2$-invariant initial data sets on simply connected four dimensional manifolds $\Sigma$. Moreover, we extend the local mass angular momenta inequality result \cite{alaee2015proof} for $U(1)^2$ invariant black holes to the case with nonzero stress energy tensor with positive matter density and energy-momentum current invariant under the above symmetries.
\end{abstract}
\maketitle

\section{Introduction} 
In \cite{brill1959positive} Brill proved a positive energy theorem for a certain class of maximal, axisymmetric initial data sets on $\mathbb{R}^3$. Brill's theorem has been extended by Dain \cite{dain2008proof} and Gibbons and Holzegel  \cite{gibbons2006positive} for a larger class of 3 dimensional initial data.   Subsequently, Chrus\'ciel\cite{chrusciel2008masspositivity} generalized the result to any maximal initial data set on a simply connected manifold (with multiple asymptotically flat ends) admitting a $U(1)$ action by isometries. Moreover, in \cite{gibbons2006positive}  a positive energy theorem was proved for a restricted class of maximal, $U(1)^2$-invariant, four-dimensional initial data sets on $\mathbb{R}^4$.  The first purpose of this note is to generalize this latter result to a larger class of 4+1 initial data. In particular, our result extends the work of  \cite{gibbons2006positive} in three main directions: 
\begin{enumerate}
\item We consider the general form of a $U(1)^2$-invariant metric (i.e. we do not assume the initial data has an orthogonally transitive $U(1)^2$ isometry group) on asymptotically flat, simply connected, four-dimensional manifolds $\Sigma$ admitting a torus action.  
\item The orbit space $\mathcal{B}\cong\Sigma/U(1)^2$ of $\Sigma$ belongs to a larger class $\Xi$ which is defined below in Definition \ref{orbitclass}.  The boundary conditions on axis and fall-off conditions at spatial infinity are weaker than those considered in \cite{gibbons2006positive}. In particular they include the data corresponding to maximal spatial slices of the Myers-Perry black hole. 
\item The manifold $\Sigma$ may possess an additional end (either asymptotically flat or asymptotically cylindrical of the form $\mathbb{R} \times S^3$). Such $\Sigma$ arise in the example of complete initial data for black hole spacetimes. The existence of non-trivial topology is also required for initial data to carry non-vanishing angular momenta.  The results also hold for data satisfying (1) and (2) on  $\mathbb{R}^4$. 
\end{enumerate}   

The second main result of this work is to extend the local mass-angular momenta inequality proved in \cite{alaee2015proof} to the non-vacuum case with positive energy density and vanishing energy current in directions tangent to the generators of the isometry group. This result naturally extends the result of \cite{dain2013lower} to the 4+1-dimensional setting. 

\section{Positivity of mass}
An asymptotically flat maximal initial data set $(\Sigma,h,K,\mu,j)$ must satisfy the Einstein constraint equations
\begin{equation}
R_h-|K|^2_h=16\pi\mu,\qquad \text{div}K=8\pi j\,.
\end{equation}
where $\mu$ is the energy density, $j$ is an energy-momentum current, and $R_h$ and $|K|^2_h$ are  respectively the Ricci scalar curvature and full contraction of $K$ with respect to $h$. $\Sigma$ is assumed to be a complete, oriented, simply connected asymptotically flat spin manifold with an additional asymptotic end.  We now briefly review the discussion\footnote{The replaced preprint version contains an improved discussion of the geometry of $\Sigma$.}  in \cite{alaee2014mass}.  As proved in \cite{hollands2011further,orlik1970actions}  if the manifold-with boundary $M$ is a spatial slice of the domain of outer communications of an asymptotically flat black hole spacetime admitting an $U(1)^2$ action, then $\Sigma \cong \mathbb{R}^4 \# n \, (S^2 \times S^2)-B$ for some integer $n$ where $B$ is a four-manifold with closure $\bar{B}$ such that $\partial{\bar{B}} = H$ and $H$ is a spatial cross section of the event horizon.  We obtain a complete manifold $\Sigma$ by doubling $M$ across its boundary $\partial M$ \cite{alaee2014mass}.  For example, complete initial data for the non-extreme Myers-Perry black hole has $\Sigma \cong \mathbb{R} \times S^3$, which has two asymptotically flat ends.  For extreme black hole initial data, a spatial slice of the domain of outer communications is already complete (the horizon is an infinite proper distance away from any point in the interior). Complete initial data for the extreme Myers-Perry black hole again has $\Sigma \cong \mathbb{R} \times S^3$, although the geometry is now cylindrical at one end.  Note that initial data for non-extreme and extreme black rings have different topology \cite{alaee2014mass}. 

We consider $U(1)^2 = U(1) \times U(1)$ invariant data with generators $\xi_{(i)}$ for $i=1,2$. $\Sigma$ is therefore equipped with a $U(1)^2$ action and further $\mathcal{L}_{\xi_{(i)}}K=\mathcal{L}_{\xi_{(i)}}h=0$. It proves useful to represent our space of functions  on the two-dimensional orbit space $\mathcal{B} \equiv \Sigma/U(1)^2$.  in general the action will have fixed points (i.e. on points where a linear combination of the $\xi_{(i)}$ vanish). A careful analysis  \cite{hollands2008uniqueness} establishes that $\mathcal{B}$ is an analytic, simply connected manifold with boundaries and corners and can be described as follows. Define the Gram matrix $\lambda_{ij} = \xi_{(i)} \cdot \xi_{(j)}$. On interior points of $\mathcal{B}$  the rank of $\lambda_{ij}$ is 2.  The boundary is divided into segments. On each such segment the rank of $\lambda_{ij}$ is one and there is an integer-valued vector $v^i$ such that $\lambda_{ij}v^j =0$ on each point of the segment (i.e. the Killing field $v^i \xi_i$ vanishes on this segment). On corner points, where adjacent boundary segments meet, the rank of $\lambda_{ij}$ vanishes. Moreover, if $\bm{v}_{s} = (v^1_s,v^2_s)^t$ and $\bm{v}_{s+1}$ are vectors associated with two adjacent boundary segments then we must have $\det(\bm{v}_s,\bm{v}_{s+1})=\pm 1$  \cite{hollands2008uniqueness}.  Finally, we note that since $\Sigma$ has two asymptotic ends, the two-dimensional orbit space is an open manifold with two ends.  Note that at interior points,  the orbit space is equipped with the quotient metric
\begin{equation}\label{orbitspacemetric}
q_{ab} = h_{ab} - \lambda^{ij} \xi_{ia} \xi_{ib}
\end{equation} 

The orbit space $\mathcal{B}$ is a simply connected, analytic two-manifold with (smooth) boundaries and corners, with two ends. By the Riemann mapping theorem, it can be analytically mapped to the upper half plane of $\mathbb{C}$ with a point removed on the real axis (if the point is removed anywhere else, then the region will not be simply connected). The boundary of $\mathcal{B}$ is mapped to the real axis with the above point removed by Osgood-Caratheodory theorem \cite{noguchi1990geometric}, which we take to be the origin without loss of generality.  We assume that \eqref{orbitspacemetric} admits the global representation
\begin{equation}
q\equiv e^{2U + 2v} ( \td \rho^2 + \td z^2)
\end{equation} where $U=U(\rho,z), v = v(\rho,z)$ are smooth functions and $\rho \in [0,\infty)$ and $z \in \mathbb{R}$. The asymptotically flat end corresponds to $\rho,z \to \infty$ and the point $(\rho,z) = (0,0)$ corresponds to the second asymptotic end.  We will impose appropriate decay conditions on $(U,v)$ below.  The boundary is characterized by $\rho =0$ in this representation. The boundary segments, where a particular linear combination of Killing fields vanish, are then described by the intervals  $I_s=\{(\rho,z) | \rho=0, a_s < z < a_{s+1}\}$ where $a_1<a_2<\cdots<a_n$ are referred to as `rod points'. Asymptotic flatness requires that there are two semi-infinite rods $I_- = \{ (\rho,z)|\rho =0, -\infty < z < a_1\}$ and $I_+ = \{(\rho,z) | \rho = 0, a_n < z < \infty\}$ corresponding to the two symmetry axes of the asymptotically flat region.  Further details on the orbit space can be found in \cite{alaee2014mass}.

Now note $\det \lambda(0,z) =0$ on corner and boundary points and smoothness at fixed points requires $\det \lambda = \rho^2 + O(\rho^4)$ as $\rho \to 0$. Furthermore since $\Sigma$ is asymptotically flat, this implies $\det \lambda$ has to approach the corresponding value in Euclidean space outside a large ball (i.e. $\det \lambda \sim r^4$ as $r \to \infty$ where $r$ is a radial coordinate in $\mathbb{R}^4$).  Let $\phi^i$ be coordinates with period $2\pi$ such that the $\mathcal{L}_{\xi_i} \phi^j = \delta^j_i$.  Then $\xi_{(i)} = \partial_{\phi^i}$. 

The four-manifold $(\Sigma,h)$ may be considered as the total space of a $U(1)^2$ principal bundle over $\mathcal{B}$, where we identify the fibre metric with $\lambda_{ij}$.   We use Greek indices $\alpha,\beta=1,...,4$ to label local coordinates on $\Sigma$.  The simplest case is $\mathbb{R}^4$ with its Euclidean metric which in our coordinate system has the representation
\begin{equation}\label{flatmet}
\delta_4= \frac{\td\rho^2 + \td z^2}{2\sqrt{\rho^2 + z^2} } + (\sqrt{\rho^2 + z^2} - z) (\td\phi^1)^2 + (\sqrt{\rho^2 + z^2} + z) (\td\phi^2)^2
\end{equation}  Asymptotically flat metrics must approach $\delta_4$ with appropriate fall-off conditions. In particular we have $\det \lambda \to \rho^2$ as $\rho,z \to \infty$.  This suggests we set $ \lambda_{ij} = e^{2v} \lambda'_{ij}$ where $\det \lambda' = \rho^2$ and $v$ satisfies appropriate decay conditions at the ends and boundary conditions on the axis.  These decay conditions are most appropriately expressed in terms of new coordinates $(r,x)$ defined by 
\begin{equation}
r \equiv \left[ 2\left[\rho^2 + z^2\right]^{1/2}\right]^{1/2} \qquad x \equiv \frac{z}{\left[\rho^2 + z^2\right]^{1/2}}
\end{equation}  where $0 \leq r < \infty$ and $-1 \leq x \leq 1$. The axis $\Gamma$ now corresponds to two lines $\mathcal{I}^+ \equiv \{(r,x)| x= 1\}$ and $\mathcal{I}^- \equiv \{(r,x)| x= -1\}$ . Note that if the space has a second asymptotic end, then the point $r=0$ is removed.  In this representation, the Euclidean metric on $\mathbb{R}^4$ takes the form
\begin{equation}\label{rxcoords}
\delta_4 =  \td r^2 + r^2 \left[ \frac{\td x^2}{4(1-x^2)} + \frac{(1-x)}{2} (\td \phi^1)^2 + \frac{(1+x)}{2} (\td \phi^2)^2\right]								
\end{equation}

We consider initial data $(\Sigma,h)$ which are a natural generalization of the well-known Brill data for three-dimensional initial data sets.  Motivated by the above discussion, we define this class as follows:

\begin{Def}[Generalized Brill data]\label{GBdata}We say that an initial data set $(\Sigma,h,K,\mu,j)$ for the Einstein
equations is a Generalized Brill (GB) initial data set with local metric

\begin{equation}
h =e^{2v} \left[e^{2U}\left(\td\rho^2+\td z^2\right)+\lambda'_{ij}\left(\td\phi^i+A^i_a\td x^a\right)\left(\td\phi^j+A^j_a\td x^a\right)\right]\label{generalmetric}
\end{equation}
where $(x^1,x^2)=(\rho,z)$,  $\det\lambda'=\rho^2$ and $U=V-\frac{1}{2}\log\left(2\sqrt{\rho^2+z^2}\right)$  if it satisfies the following conditions.
\begin{enumerate}[leftmargin=*]
\item  $(\Sigma,h)$ is a simply connected Riemannian manifold and $M_{\text{end}}$ is diffeomorphic to $\mathbb{R}^4\,\backslash \overline{B_R(0)}$ where $B_R(0)$ is an open ball with large radius $R$ such that\footnote{This condition is asymptotically flatness \cite{bartnik1986mass}  for $s\geq 2$ and when we write $f=o_s(r^{l})$ it means  $\partial_{\beta_1}\cdots\partial_{\beta_p}f=o(r^{l-p})$ for $0\leq p\leq s$.}
\begin{equation}
h_{\alpha\beta}-(\delta_4)_{\alpha\beta}=o_s(r^{-1}),\quad \partial h\in L^2(M_{\text{end}})\quad \mu,j\in L^1(M_{\text{end}}),\quad K=o_{s-1}(r^{-2})\nonumber
\end{equation}
\item The second fundamental form satisfies
\begin{equation}
\mathcal{L}_{\xi_{(i)}}K=\mathcal{L}_{\xi_{(i)}}h=0\qquad \text{Tr}_{h}K=0\nonumber
\end{equation} (i.e. the data is maximal).
\item The coordinate system $(\rho,z,\phi^i)$ forms a global coordinate system\footnote{It may be possible to prove this assumption is unnecessary  (see \cite{chrusciel2008masspositivity} for the three-dimensional case)} on $\Sigma$ where $\rho \in \mathbb{R}^+ \cup \{0\}$, $z \in \mathbb{R}$, and $\phi^i$ have period $2\pi$. I  The functions $v,V,A^i_a,$ and $\lambda'_{ij}$  satisfy the following decay conditions, which are best expressed in terms of the $(r,x)$ chart given by \eqref{rxcoords}:
\begin{enumerate}[leftmargin=*]
\item \label{Asympvf} as $r\to\infty$  
\begin{gather}
v=o_{1}(r^{-1}),\quad A_{\rho}^i=\rho o_{1}(r^{-5}),\quad A_z^i=o_{1}(r^{-3}),\quad V=o_1(r^{-1})\nonumber\\
\lambda'_{ii}=\left(1+(-1)^{i-1}f_{11}r^{-1-\kappa}+o_1(r^{-2})\right)\sigma_{ii},\quad \lambda'_{12}=\rho^2o_1(r^{-5}),\nonumber
\end{gather}
where $0<\kappa\leq 1$,  $\sigma_{ij}=\frac{r^2}{2}\text{diag}\left(1-x,1+x\right)$
\item\label{End1} If $r\to 0$  represents a second asymptotically flat end we have
\begin{gather}
v=-2\log (r)+o_{1}(r),\quad A_{\rho}^i=\rho o_{1}(r),\quad A_z^i=o_{1}(r^{3}),\quad V=o_1(r)\nonumber\\
\lambda'_{ii}=\left(1+(-1)^{i-1}f_{22}r^{1+\kappa}+o_1(r^{2})\right)\sigma_{ii},\quad \lambda'_{12}=\rho^2o_1(r^{-1}),\nonumber
\end{gather}
\item\label{End2} If $r\to 0$ is a cylindrical end with topology $\mathbb{R}^+\times N$ where $N\cong S^3,S^1\times S^2, L(p,q)$ we have
\begin{gather}
v=-\log (r)+O_{1}(r^{1}),\quad A_{\rho}^i=\rho o_{1}(r),\quad A_z^i=o_{1}(r^{3}),\nonumber\\
\lambda'_{ij}-r^2\bar{\sigma}_{ij}=o_1(r^{2}),\quad
V=O_1(1)\,.\nonumber
\end{gather}
where $h^c=e^{2V}\frac{dx^2}{4(1-x^2)}+\bar{\sigma}_{ij}\td\phi^i\td\phi^j$ is a metric on $N$.
\item \label{Vaxis} as $\rho\to 0$ and $\bm{w}=w^i\frac{\partial}{\partial\phi_i}$ is the Killing vector vanishes on the rod $I_s$
\begin{gather}
\lambda'_{ij}w^j=O(\rho^2),\qquad \text{and others $\lambda'_{ij}=O(1)$}\,.\nonumber
\end{gather}
and to avoid conical singularities on the axis $\Gamma$  we have
\begin{equation}
V(z)=\frac{1}{2}\lim_{\rho\to 0}\log\left(\frac{2\sqrt{\rho^2+z^2}\lambda'_{ij}w^iw^j}{\rho^2}\right)\equiv\frac{1}{2}\log V_s,\quad z\in I_s=(a_s,a_{s+1}),\quad w^i\in\mathbb{Z}\,.\nonumber
\end{equation}
\end{enumerate}
\end{enumerate}
\end{Def} 

We remark that any sufficiently smooth, asymptotically flat metric on a simply connected 3-manifold with additional asymptotic ends obtained by removing points form $\mathbb{R}^3$ and admitting a $U(1)$ isometry can be written in the above form, with $i=1$ \cite{chrusciel2008masspositivity}.  It is natural to expect a similar result holds in the present case, up to some additional conditions.  Note that the one-forms $A^i = A^i_a \td x^a$ may be considered as a local connection on the $U(1)^2$ bundle over $\mathcal{B}$.

The initial data sets defined above encompass a large class of possible data sets, which include in particular initial data for extreme and non-extreme black rings. It proves useful to restrict attention to a subclass of data, which includes initial data for the Myers-Perry black hole.  Let a fixed GB data set have orbit space $\mathcal{B}$ with rod points $a_1 , a_2 \ldots a_n$. Via the transformation \eqref{rxcoords} these points map to $\mathcal{I}^+$ and $\mathcal{I}^-$. We arrange these points in order of increasing $r$ and denote by $b_s$, for $s=1 \ldots n'\leq n$, with $b_0=\mathcal{I}_E\equiv\{(r,x):r=0, -1\leq x\leq 1\}$ and $b_{n'+1}=\mathcal{I}_F\equiv\{(r,x):r=\infty, -1\leq x\leq 1\}$. The $\mathcal{I}_F$ is the asymptotically flat end and $\mathcal{I}_E$ is another asymptotic end or just the origin of half plan $(\rho,z)$.

\begin{Def}\label{orbitclass}The \emph{admissible set} $\Xi$ of orbit spaces is a collection of $\mathcal{B}$ such that \emph{distinct} rotational Killing fields vanish on $\Gamma \cap B_s$, where $B_s = \{ (r,x)| e^{b_{s}} \leq r \leq e^{b_{s+1}}, -1 \leq x \leq 1\}$. 
\end{Def}
\begin{remark} The regions $B_s$ correspond to annuli in the $(\rho,z)$ representation of $\mathcal{B}$ and (finite, infinite, or semi-infinite) rectangles on the $(y,x)$ representation where $y = \log r$. 
\end{remark}
\begin{remark} The geometry of a second asymptotic end of data belonging to $\Xi$ must have $N = S^3$ (or a Lens space quotient).  This follows from the classification of orbit spaces $N / U(1)^2$ obtained in  \cite{hollands2008uniqueness}  when distinct Killing fields vanish on $\mathcal{I}^+$ and $\mathcal{I}^-$. 
\end{remark}
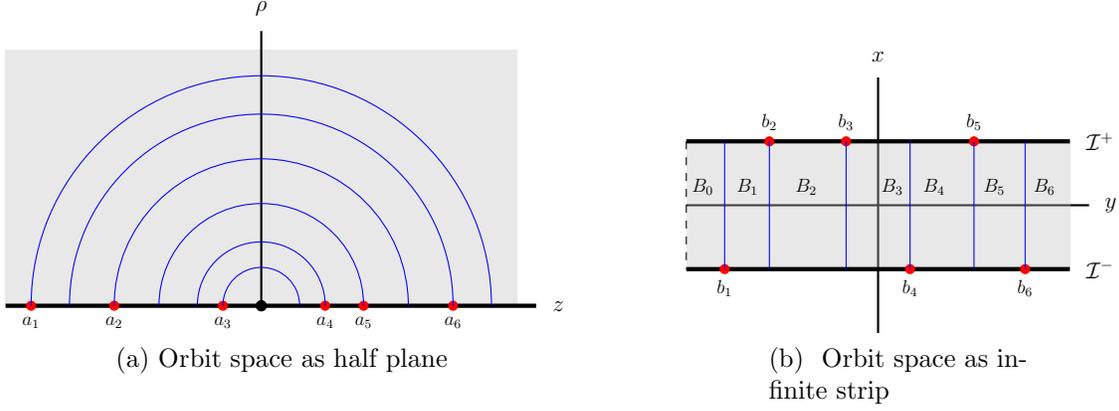
\begin{figure}[h]
\centering
\subfloat[{Orbit space as half plane}]{
\begin{tikzpicture}[scale=.85, every node/.style={scale=0.6}]
\fill[ fill=gray!50!white, opacity=0.35,very thick](-4,4)--(-4,0)--(4,0)--(4,4);
\draw[black,ultra thick](-4,0)--(4.3,0)node[black,font=\large,right=.2cm]{$z$};
\draw[black,thick](0,0)--(0,4.3)node[black,font=\large,above=.2cm]{$\rho$};
\draw[red,fill=red] (-3.6,0) circle [radius=.07] node[black,below=.1cm]{$a_1$};
\draw[red,fill=red] (-2.3,0) circle [radius=.07] node[black,below=.1cm]{$a_2$};
\draw[red,fill=red] (-.6,0) circle [radius=.07] node[black,below=.1cm]{$a_{3}$};
\draw[fill=black] (0,0) circle [radius=.08];
\draw[red,fill=red] (1,0) circle [radius=.07] node[black,below=.1cm]{$a_{4}$};
\draw[red,fill=red] (1.6,0) circle [radius=.07] node[black,below=.1cm]{$a_{5}$};
\draw[red,fill=red] (3,0) circle [radius=.07] node[black,below=.1cm]{$a_{6}$};
\draw[blue] (3,0) arc (0:180:3cm);
\draw[blue] (3.6,0) arc (0:180:3.6cm);
\draw[blue] (2.3,0) arc (0:180:2.3cm);
\draw[blue](.6,0) arc (0:180:.6cm);
\draw[blue] (1,0) arc (0:180:1cm);
\draw[blue](1.6,0) arc (0:180:1.6cm);
\end{tikzpicture}}
\hspace{3em}
\subfloat[{ Orbit space as infinite strip }]{
\begin{tikzpicture}[scale=.85, every node/.style={scale=0.6}]
\draw[black,thick](-3,0)--(3.3,0)node[black,font=\large,right=.2cm]{$y$};
\draw[black,thick](0,-2)--(0,2)node[black,font=\large,above=.2cm]{$x$};
\fill[gray!50!white, opacity=0.35](-3,1)--(3,1)--(3,-1)--(-3,-1);
\draw[black,ultra thick](-3,1)--(3,1)node[black,font=\large,right=.2cm]{$\mathcal{I}^+$};
\draw[black,ultra thick](-3,-1)--(3,-1)node[black,font=\large,right=.2cm]{$\mathcal{I}^-$};
\draw[red,fill=red] (-2.4,-1) circle [radius=.07] node[black,below=.1cm]{$b_1$};
\draw[red,fill=red] (.5,-1) circle [radius=.07] node[black,below=.1cm]{$b_4$};
\draw[red,fill=red] (2.3,-1) circle [radius=.07] node[black,below=.1cm]{$b_6$};
\draw[red,fill=red] (-1.7,1) circle [radius=.07] node[black,above=.1cm]{$b_2$};
\draw[red,fill=red] (-.5,1) circle [radius=.07] node[black,above=.1cm]{$b_3$};
\draw[red,fill=red] (1.5,1) circle [radius=.07] node[black,above=.1cm]{$b_5$};
\draw[blue](-2.4,-1)--(-2.4,1)node[black,below=1cm,left=.1cm]{$B_{0}$};
\draw[blue](.5,-1)--(.5,1)node[black,below=1cm,left=.01cm]{$B_{3}$};
\draw[blue](2.3,-1)--(2.3,1)node[black,below=1cm,left=.3cm]{$B_{5}$}node[black,below=1cm,right=.04cm]{$B_{6}$};
\draw[dashed](-3,1)--(-3,-1);
\draw[blue](-1.7,-1)--(-1.7,1)node[black,below=1cm,left=.1cm]{$B_{1}$};
\draw[blue](-.5,-1)--(-.5,1)node[black,below=1cm,left=.5cm]{$B_{2}$};
\draw[blue](1.5,-1)--(1.5,1)node[black,below=1cm,left=.5cm]{$B_{4}$};
\end{tikzpicture}}
\caption{The orbit space can be subdivided into subregions $B_s$ which are half-annuli in the $(\rho,z)$ plane and  rectangles in the $(y,x)=(\log r,x)$ plane. In this case $n=6$. The dashed line $\mathcal{I}_E$ can represent origin or in the case of black holes is another asymptotic end.}\label{fig5}
\end{figure} 
\noindent The ADM energy\footnote{We will refer to this as the `mass' hereafter.} and momenta for a generalized Brill data set $(\Sigma,h,K,\mu,j)$ are given by
\begin{equation}\label{ADMmass}
m=\frac{1}{16\pi}\lim_{r\rightarrow\infty}\int_{S^3_r}\left(\partial_{\alpha}h_{\alpha\beta}-\partial_{\beta}h_{\alpha\alpha}\right)n^{\beta}\, ds,\qquad\quad j_{\alpha}=\frac{1}{8\pi}\lim_{r\rightarrow\infty}\int_{S^3_r}K_{\alpha\beta}n^{\beta}\, ds
\end{equation}
 where $S^3_r$ refers to a three-sphere of coordinate radius $r$ with volume element $ds=\frac{r^3}{4}\td x\td \phi^1\td\phi^2$ in the Euclidean chart outside a large compact region and $n$ is the unit normal. Then we have the following positive mass theorem.

\begin{thm}[Positive Mass Theorem for $U(1)^2$-invariant data]\label{positive mass}
Consider a GB initial data set $(\Sigma,h,K,\mu,j)$. Then if $R_{h}\geq 0$ and $\mathcal{B}\in \Xi$ where $\Xi$ is defined in Definition \ref{orbitclass}, then 
\begin{equation}
0\leq m\leq \infty\,.
\end{equation}
Moreover, we have $m<\infty$ if and only if we have 
\begin{equation}
R_{h},\rho^{-2}\det\nabla\lambda'\in L^1(\mathcal{B}),\quad V\in L^1(\mathbb{R}^+),\quad (A^i_{\rho,z}-A^i_{z,\rho}),v\in L^2(\mathcal{B})\label{bound}
\end{equation}Finally, 
 $m=0$ if and only if $h$ is the Euclidean metric and $\Sigma=\mathbb{R}^4$. 
\end{thm}
\begin{proof}
Consider the GB data $(\Sigma,h,K,\mu,j)$. We can write the metric in conformal form as 
\begin{equation}\label{conformal}
\tilde{h}_{\alpha\beta}=\Phi^{-2}h_{\alpha\beta}\,.
\end{equation}
where $\Phi=e^{v}$.  Then by the asymptotic decay properties of GB data at the asymptotically flat end we have  
\begin{eqnarray}
\Phi-1=o(r^{-1}),\qquad \tilde{h}_{\alpha\beta}=\delta_{\alpha\beta}+o(r^{-1})\,.\label{fall2}
\end{eqnarray}
Then the integrand in the expression for the ADM mass  \eqref{ADMmass} is 
\begin{eqnarray}
\partial_{\alpha} h_{\alpha\beta}-\partial_{\beta} h_{\alpha\alpha}&=&-6\partial_{\beta}\Phi+\partial_\alpha\tilde{h}_{\alpha\beta}-\partial_\beta\tilde{h}_{\alpha\alpha}+o(r^{-3}).
\end{eqnarray}
Therefore we find 
\begin{eqnarray}\label{ADMmassofh}
m&=&\frac{1}{16\pi}\lim_{r\rightarrow\infty}\int_{S_r}-6\partial_c\Phi\, n^c\, ds+m_{\tilde{{h}}}\nonumber\\
&=&\frac{1}{64\pi}\lim_{r\rightarrow\infty}\int_{S_r}-6v_{,r}\, r^3\td x\td\phi^1\td\phi^2+m_{\tilde{{h}}}\\
&=&-\frac{3\pi}{8}\int_{\mathcal{I}_F}v_{,r}\, dx+m_{\tilde{{h}}}
\end{eqnarray}
where  we used $\Phi=e^{v}=1+o_1(r^{-1})$ as $r\to\infty$ in first equality. The second equality follows from $U(1)^2$-invariant symmetry of $v$ and definition of $\mathcal{I}_F=\{(r,x):r=\infty, -1\leq x\leq 1\}$. Now we find the ADM mass of the conformal metric $\tilde{h}$.
\begin{lemma}
Consider a GB data $(\Sigma,{h},K,\mu,j)$ with the rescaling \eqref{conformal}. Then
\begin{equation}
m_{\tilde{{h}}}=-\frac{\pi}{4}\int_{\mathcal{I}_F}\left(\frac{r^3}{2}V_{,r}-r^2V\right)\td x
\end{equation} 
\end{lemma}
\begin{proof}
Consider the flat metric in Cartesian coordinates $(y_i)$ 
\begin{equation}
\delta_4=\td y_1^2+\td y_2^2+\td y_3^2+\td y_4^2.
\end{equation}
with the transformation to $(r,x,\phi^1,\phi^2)$ (equivalently $(\rho=\frac{r^2}{2}\sqrt{1-x^2},z=\frac{r^2}{2}x,\phi^1,\phi^2)$) for GB conformal metric with transformation
\begin{eqnarray}
y_1&=&r\sqrt{\frac{1+x}{2}}\cos\phi^1\quad y_2=r\sqrt{\frac{1+x}{2}}\sin\phi^1,\\ y_3&=&r\sqrt{\frac{1-x}{2}}\cos\phi^2,\quad y_4=r\sqrt{\frac{1-x}{2}}\sin\phi^2,\nonumber
\end{eqnarray}
First we write the conformal metric in the $(r,x,\phi^1,\phi^2)$ chart:
\begin{eqnarray}
\qquad\quad\tilde{{h}}&=&\delta_4+\underbrace{\left(e^{2V}-1\right)\left(\td r^2+\frac{r^2}{4(1-x^2)}\td x^2\right)}_{B_I}+\underbrace{\left(\lambda'_{ij}-\sigma_{ij}\right)\td\phi^i\td\phi^i}_{B_{II}}+2\underbrace{\lambda'_{ij}A^i_a\td x^a\td\phi^j}_{B_{III}}\label{conformalexpansion}\\
&+&\text{terms quadratic in $A^i_a$}.\nonumber
\end{eqnarray}
The mass of $\delta_4$ is zero. By part 2 in Definition \ref{GBdata}, the last quadratic terms in \eqref{conformalexpansion} does not contribute to the mass integral. Now we compute the mass of the terms $B_I, B_{II}$, and $B_{III}$. By asymptotic behaviour of functions (part (3) of Definition \ref{GBdata}) we have
\begin{eqnarray}
B_I+B_{II}&=&\left(e^{2V}-1\right)\delta_4+\frac{1}{2}\left[\frac{f_{11}}{r^2}-\left(e^{2V}-1\right)\right]r^2(1+x)\left(\td\phi^1\right)^2\\
&+&\frac{1}{2}\left[-\frac{f_{11}}{r^2}-\left(e^{2V}-1\right)\right]r^2(1-x)\left(\td\phi^2\right)^2\nonumber\\
&=&\underbrace{\left(e^{2V}-1\right)}_{C_I}\delta_4+\underbrace{\left[\frac{f_{11}}{r^2}-\left(e^{2V}-1\right)\right]\frac{(y_1\td y_2-y_2\td y_1)^2}{y_1^2+y_2^2}}_{C_{II}}\nonumber\\
&+&\underbrace{\left[-\frac{f_{11}}{r^2}-\left(e^{2V}-1\right)\right]\frac{(y_4\td y_3-y_3\td y_4)^2}{y_3^2+y_4^2}}_{C_{III}}.\nonumber
\end{eqnarray} 
We compute the ADM mass of each one of these terms :
\begin{itemize}[leftmargin=*]
\item $C_{I}:$ This is  a conformally flat metric and by \eqref{ADMmassofh} we obtain
\begin{equation}
m^{C_I}=\frac{1}{16\pi}\lim_{r\rightarrow\infty}\int_{S_r}-3\partial_r\left(e^{2V}-1\right)\,\td s.
\end{equation}
\item $C_{II}:$ we consider $C_{II}$ as a metric $(C_{II})_{ab}$ such that  the only nonzero components are
\begin{equation}
(C_{II})_{ab}=\frac{1}{y_1^2+y_2^2}\left[\frac{f_{11}}{r^2}-\left(e^{2V}-1\right)\right]\left(y_2^2\td y_1^2+y_1^2\td y_2^2-2y_1y_2\td y_1\td y_2\right).
\end{equation}
\noindent Then by definition of ADM mass \eqref{ADMmass} we obtain 
\begin{eqnarray}
m^{C_{II}}\noindent&=&\frac{1}{16\pi}\lim_{r\rightarrow\infty}\int_{S_r}\left[\partial_{y_1}(C_{II})_{y_1y_2}\frac{y_2}{r}+\partial_{y_2}(C_{II})_{y_1y_2}\frac{y_1}{r} + \partial_{y_1}(C_{II})_{y_1y_1}\frac{y_1}{r} \right.\nonumber \\
 &+&\left. \partial_{y_2}(C_{II})_{y_2y_2}\frac{y_2}{r} -  \partial_{y_i}(C_{II})_{y_1y_1}\frac{y_i}{r}  -\partial_{y_i}(C_{II})_{y_2y_2}\frac{y_i}{r}\right]\,\td s\nonumber\\ 
&=&-\frac{1}{16\pi}\lim_{r\rightarrow\infty}\int_{S_r}\left\{\partial_r\left[\frac{f_{11}}{r^2}-\left(e^{2V}-1\right)\right]+\frac{1}{r}\left[\frac{f_{11}}{r^2}-\left(e^{2V}-1\right)\right]\right\}\,\td s.
\end{eqnarray}
\item $C_{III}:$ This is similar to $C_{II}$ and we have
\begin{equation}
m^{C_{III}}=-\frac{1}{16\pi}\lim_{r\rightarrow\infty}\int_{S_r}\left\{\partial_r\left[-\frac{f_{11}}{r^2}-\left(e^{2V}-1\right)\right]+\frac{1}{r}\left[-\frac{f_{11}}{r^2}-\left(e^{2V}-1\right)\right]\right\}\,\td s.
\end{equation}
\end{itemize}
Hence the ADM mass of $B_{I}+B_{II}$ is
\begin{eqnarray}
m^{B_{I}+B_{II}}&=&\frac{1}{16\pi}\lim_{r\rightarrow\infty}\int_{S_r}\left\{-\partial_r\left(e^{2V}-1\right)+\frac{2}{r}\left(e^{2V}-1\right)\right\}\,\td s\\
&=&\frac{1}{16\pi}\lim_{r\rightarrow\infty}\int_{S_r}\left\{-\partial_r\left(e^{2V}-1\right)+\frac{2}{r}\left(e^{2V}-1\right)\right\}\,\frac{r^3}{4}\td x\td \phi^1\td\phi^2 \\ 
&=&\frac{\pi}{4}\lim_{r\rightarrow\infty}\int_{S_r}\left\{-2V_{,r}+\frac{4V}{r}\right\}\,\frac{r^3}{4}\td x=-\frac{\pi}{4}\int_{\mathcal{I}_F}\left(\frac{r^3}{2}V_{,r}-r^2V\right)\td x.\nonumber
\end{eqnarray}
where in the second line we used part (3)-a of Definition \ref{GBdata}. We consider the term $B_{III}$ 
\begin{equation}
B_{III}=\underbrace{\frac{1}{2}r^2(1+x)\td\phi^1\left(A^1_\rho\td\rho+A^1_z\td z\right)}_{D_I+D_{II}}+\underbrace{\frac{1}{2}r^2(1-x)\td\phi^2\left(A^2_\rho\td\rho+A^2_z\td z\right)}_{D_{III}+D_{IV}}+o(r^{-3}).
\end{equation}
We prove ADM mass of the $D_I$ and $D_{II}$ parts are zero and the argument for the other terms are similar.  As in the argument used for $C_{II}$ and $C_{III}$, we consider $D_I$ as the following metric
\begin{eqnarray}
(D_I)_{ab}&=&\frac{1}{2}r^2(1+x)\td\phi^1A^1_\rho\td\rho=\left(y_1\td y_2-y_2\td y_1\right)A^1_\rho\td\sqrt{\left(y_1^2+y_2^2\right)\left(y_3^2+y_4^2\right)}\\
&=&\frac{A^1_\rho}{\rho}(y_3^2+y_4^2)(y_1\td y_2-y_2\td y_1)(y_1\td y_1+y_2\td y_2)\nonumber\\
&+&\frac{A^1_\rho}{\rho}(y_1^2+y_1^2)(y_1\td y_2-y_2\td y_1)(y_3\td y_3+y_3\td y_3).\nonumber
\end{eqnarray}
Then the integrand appearing in the ADM mass expression is 
\begin{eqnarray}
\left(\partial_a(D_I)_{ac}-\partial_c(D_I)_{aa}\right)n^c=\rho\underbrace{(y_1\partial_{y_2}-y_2\partial_{y_1})A^1_{\rho}}_{=0}=0.
\end{eqnarray}
Now consider $D_{II}$ as a metric
\begin{eqnarray}
(D_{II})_{ab}&=&\frac{1}{2}r^2(1+x)\td\phi^1A^1_z\td z=\frac{1}{2}\left(y_1\td y_2-y_2\td y_1\right)A^1_z\td\left[\left(y_1^2+y_2^2\right)-\left(y_3^2+y_4^2\right)\right]\nonumber\\
&=&\frac{A^1_z}{z}(y_1\td y_2-y_2\td y_1)(y_1\td y_1+y_2\td y_2)-\frac{A^1_z}{z}(y_1\td y_2-y_2\td y_1)(y_3\td y_3+y_3\td y_3).\nonumber
\end{eqnarray}
Then the ADM mass is 
\begin{eqnarray}
\left(\partial_a{D_{II}}_{ac}-\partial_c(D_{II})_{aa}\right)n^c=\frac{z}{2}\underbrace{(y_1\partial_{y_2}-y_2\partial_{y_1})A^1_{z}}_{=0}=0.
\end{eqnarray}
Therefore, the ADM mass of the conformal metric is zero, that is $m_{\tilde{h}}=0$.
\end{proof}

 Returning to the  mass of GB data  we have
\begin{eqnarray}
m&=& \frac{\pi}{4}\int_{\mathcal{I}_F}\left[-\frac{3}{2}v_{,r}-\left(\frac{r^3}{2}V_{,r}-r^2V\right)\right]\td x\,,\label{massonx}
\end{eqnarray}
Then we define three one-form $\omega$, $\chi_1$ and $\chi_2$
\begin{equation}
\omega=2\chi_1+6\chi_2
\end{equation}
where 
\begin{eqnarray}
\chi_1&=&(\rho V_{,\rho}-V)\td z-\rho V_{,z}\td\rho\\
&=&\left(-r(1-x^2)V_{,x}-rxV\right)\td r+\left(\frac{r^3}{4}V_{,r}-\frac{r^2}{2}V\right)\td x\\
\chi_2&=&\rho \left(v_{,\rho}\td z-v_{,z}\td\rho\right)=-r(1-x^2)v_{,x}\td r+\frac{r^3}{4}v_{,r}\td x
\end{eqnarray}
Then 
\begin{equation}
\td\chi_1=\Delta_2V\rho\td\rho\td z,\quad\td\chi_2=\Delta_3v\rho\td\rho\td z,\quad \td \omega=\left(2\Delta_2V+6\Delta_3v\right)\rho\td\rho\td z,\label{omegafomr}
\end{equation}
where $\Delta_3$ is Laplace operator respect to $\delta=\td\rho^2+\td z^2+\rho^2\td\phi^2$ be metric on $\mathbb{R}^3$ and $\Delta_2=\partial^2_{\rho}+\partial^2_{z}$. Now by asymptotes of GB data set, we list the behaviour of $\chi_1$ and $\chi_2$ at boundary of the orbit space $\partial\mathcal{B}=\Gamma\cup\mathcal{I}_F\cup\mathcal{I}_E$ where $\Gamma=\mathcal{I}^+\cup\mathcal{I}^-$.
\begin{eqnarray}
\chi_1&=& \left(\frac{r^3}{4}V_{,r}-\frac{r^2}{2}V\right)\td x,\quad \chi_2=\frac{r^3}{4}v_{,r}\td x,\qquad \text{on $\mathcal{I}_F$}\\
\chi_1&=&-rxV\td r,\quad \chi_2=0,\qquad \text{on $\Gamma$}\\
\chi_1&=&\chi_2=0,\qquad \text{on $\mathcal{I}_E$},
\end{eqnarray}
Now if we integrate equation\eqref{omegafomr} with coefficient $\frac{\pi}{4}$ over the orbit space $\mathcal{B}$ we have
\begin{eqnarray}
\frac{\pi}{4}\int_{\mathcal{B}}\td\omega&=&\frac{\pi}{4}\int_{\partial\mathcal{B}}\omega\label{ADMmass1}\\
&=&-\frac{\pi}{4}\int_{\Gamma}rxV\td r+\frac{\pi}{2}\int_{\mathcal{I}_F}\left[\left(\frac{r^3}{2}V_{,r}-r^2V\right)+\frac{3r^3}{2}v_{,r}\right]\td x\nonumber\\
&=&\frac{\pi}{2}\int_{0}^{\infty}r\left[V(x=1)+V(x=-1)\right]\td r-m\nonumber
\end{eqnarray}
The first equality follows from Stokes theorem and the last equality follows from equation \eqref{massonx} and orientation of $(r,x)$ chart.
We next compute the scalar curvature of $\tilde{h}_{\alpha\beta}$. After a  conformal rescaling we have
\begin{equation}
-R_{h}e^{2v}=-R_{\tilde{h}}+6e^{-2U}\left[\Delta_{3}v+\abs{\nabla v}^2\right]\,.\label{curv1}
\end{equation}
where $\nabla$ is the derivative with respect to $\delta_{ab}$ and $R_{\tilde{h}}$ is Ricci scalar of ${\tilde{h}}$. Now similar to the calculation in \cite{harmark2009domain} we compute\footnote{There is a typo in equation A.1 of \cite{harmark2009domain} and the correct expression yields \eqref{ricci} for $\textrm{Ric}(\tilde{h})$} the Ricci tensor of $\tilde{h}_{\alpha\beta}$:
\begin{eqnarray}
\qquad\tilde{R}_{ij}&=& -\frac{1}{2}\nabla_a\nabla^a\lambda'_{ij}-\frac{1}{2}\nabla_a\left(\log\rho\right)\nabla^a\lambda'_{ij}+\frac{1}{2}\nabla^a\lambda'_{ik}\lambda'^{kl}\nabla_a\lambda'_{lj}+\frac{1}{4}e^{-4U}\lambda'_{ik}\lambda'_{jl}H^{kl}\label{ricci}\\
\tilde{R}_{ia}&=&\tilde{R}_{ij}A^{j}_a+\frac{1}{2\rho}\delta_{ab}\nabla_c\left(\rho e^{-2U}\lambda'_{ij}\delta^{bd}\delta^{ce}F^j_{de}\right)\\
\tilde{R}_{ab}&=&-\tilde{R}_{ij}A^{i}_aA^{j}_b+\tilde{R}_{ia}A^i_b+\tilde{R}_{ib}A^i_a-\frac{1}{2}e^{-2U}\delta^{cd}\lambda'_{ij}F^{i}_{ac}F^{j}_{bd}+{}^2\tilde{R}_{ab}-D_aD_b\log\rho \\
&-&\frac{1}{4}\text{Tr}\left[\lambda'^{-1}\nabla_a\lambda'\lambda'^{-1}\nabla_b\lambda'\right]\nonumber
\end{eqnarray}
where $F^i_{ab}\equiv 2 \nabla_{[a}A^{i}_{b]}$,  $H^{ij}\equiv\delta^{ac}\delta^{bd}F^i_{ab}F^j_{cd}$. Here  $D_a$ and ${}^2\tilde{R}_{ab}$ are the Levi-Civita connection and Ricci tensor with respect to $q_{ab}=e^{2U}\delta_{ab}$. Then the scalar curvature is
\begin{equation}
{R}_{\tilde{h}}e^{2U}=-\frac{1}{4}e^{-2U}\lambda'_{ij}H^{ij}-2\Delta_2U+\frac{\det\nabla\lambda'}{2\rho^2}\label{curvconformal}
\end{equation}
By equations \eqref{curv1} and \eqref{curvconformal} we have
\begin{equation}
-R_{h}e^{2v+2U}=\frac{1}{4}e^{-2U}\lambda'_{ij}H^{ij}+2\Delta_2V-\frac{\det\nabla\lambda'}{2\rho^2}+6\Delta_{3}v+6\abs{\nabla v}^2\label{scalarcurv}
\end{equation}
where $\Delta_2U=\Delta_2V$ by definition of $U$. Now we integrate equation \eqref{scalarcurv} over $\mathcal{B}$ and use \eqref{ADMmass1}
\begin{eqnarray}
\quad m&=&\frac{\pi}{4}\int_{\mathcal{B}}\Bigg[R_he^{2v+2U}+\frac{1}{4}e^{-2U}\lambda'_{ij}H^{ij}-\frac{\det\nabla\lambda'}{2\rho^2}+6\abs{\nabla v}^2\Bigg]\, \td\mu\nonumber\\
&+&\frac{\pi}{2}\int_{0}^{\infty}r\left[V(x=1)+V(x=-1)\right]\td r\label{ADMGB}\\
&\geq&\frac{\pi}{4}\int_{\mathcal{B}}\Bigg[-\frac{\det\nabla\lambda'}{2\rho^2}+6\abs{\nabla v}^2\Bigg]\, \td\mu+\frac{\pi}{2}\int_{0}^{\infty}r\left[V(x=1)+V(x=-1)\right]\td r\label{positivepart}
\end{eqnarray}
The inequality follows from $H^{ij},R_h\geq 0$. Now we use the argument of Section 5 of  \cite{alaee2014mass} to establish positivity of $m$ over each annulus $B_s$. Fix $B_s$ and without loss of generality we can select the following parameterization of the 3 independent functions contained in $\lambda'_{ij}$ and $v$:
\begin{eqnarray}
\begin{aligned}\label{lambdaGH}
\lambda'_{11}&=\frac{r^2(1-x)}{2\sqrt{1-(W^s)^{2}}}e^{V^s_1-V^s_2}\qquad
\lambda'_{22}=\frac{r^2(1+x)}{2\sqrt{1-(W^s)^2}}e^{V^s_2-V^s_1},\\
\lambda'_{12}&=\frac{r^2\sqrt{1-x^2}W^s}{2\sqrt{1-(W^s)^2}}\qquad v=\frac{V^s_1+V^s_2+\log\sqrt{1-(W^s)^2}}{2}\,.
\end{aligned} 
\end{eqnarray} where $\bm{v}_s = \partial_{\bar{\phi}^1_s}$ and  $\bm{w}_s = \partial_{\bar{\phi}^2_s}$  vanish on $\mathcal{I}^+\cap B_s$ and $\mathcal{I}^-\cap B_s$, respectively such that
\begin{equation}
\frac{\partial}{\partial_{\bar{\phi}^k_s}}=\alpha^j_{sk}\frac{\partial}{\partial_{{\phi}^j}},\qquad k,j=1,2,\quad s=1,\ldots,n',\qquad \alpha^j_{sk}\in\mathbb{Z}\,.\label{changeofcoordinate}
\end{equation}
where for fixed $s$ we have $\det(\alpha^j_{sk})=\det\begin{pmatrix} \alpha^1_{s1} & \alpha^2_{s1}\\ \alpha^1_{s2} & \alpha^2_{s2}\end{pmatrix}=\pm 1$ \cite{hollands2008uniqueness}.  Recall that this relation must hold between two bases that generate the $U(1)^2$ action.  The functions $V^s_1,V^s_2$ and $W^s$ are $C^1$ functions whose boundary conditions on the axis are induced from those of $\lambda'_{ij}$ and $v$ in Definition \ref{GBdata}.  In particular, we have $\det\lambda'=\rho^2$ and to remove conical singularities on $\mathcal{I}^{\pm}$  by Definition \ref{GBdata}-\eqref{Vaxis} we require:
\begin{equation}
2V-V^s_1+V^s_2= 0 \quad \text{on $\mathcal{I}^{+}$},\quad 2V-V^s_2+V^s_1= 0 \quad \text{on $\mathcal{I}^{-}$},\quad W^s=0\quad \text{on $\mathcal{I}^{\pm}$}\label{GR1}
\end{equation} Note that since $\lambda'_{ij}$ and $v$ are continuous across the boundary of $B_s$, this will impose boundary conditions on the parameterization functions in adjacent subregions. Then we have
{\small\begin{eqnarray}
\quad m_s&\geq&\frac{\pi}{16}\int_{B_s}\Bigg(\abs{\nabla V^s_1+\nabla V^s_2}^2+\abs{\nabla V^s_1}^2+\frac{\abs{\nabla W^s}^2}{2(1-(W^s)^2)}+\frac{2(W^s)^2\abs{\nabla W^s}^2}{(1-(W^s)^2)^2}\label{positive}\\
&+&\frac{(W^s)^2}{2(1-(W^s)^2)}\left[\abs{\nabla V^s_1-\nabla V^s_2}^2-\frac{6}{W^s}(\nabla V^s_1\cdot\nabla W^s+\nabla V^s_2\cdot\nabla W^s)\right]+\abs{\nabla V^s_2}^2\nonumber\\
&+&\frac{(W^s)^2}{r^2(1-(W^s)^2)}\left[4\partial_xV^s_2-4\partial_xV^s_1+\frac{2}{(1-x^2)}\right]\Bigg)\,r^3 \td x\td r\geq 0\,.\nonumber
\end{eqnarray}} The final inequality follows from \cite{alaee2014mass,alaee2014thesis} (see also \cite{gibbons2006positive}).  The total ADM mass $m$ is simply the sum of $m_s$ and is hence non-negative. 

For the second part of the Theorem \ref{positive mass} it is obvious from \eqref{ADMGB} that  $m<\infty$ if and only if \eqref{bound} holds. Now if we assume $h$ is the Euclidean metric on $\Sigma=\mathbb{R}^4$ , clearly $m=0$. Conversely, If $m=0$, then by \eqref{ADMGB} we have
\begin{equation}
R_{{h}}=A^i_{\rho,z}-A^i_{z,\rho}=0\label{Airhoz}
\end{equation}
Now we need to show $v=0$, $V=0$, and $\lambda'_{ij}=\sigma_{ij}=\frac{r^2}{2}\text{diag}(1+x,1-x)$. We prove it by the technique we used to prove positivity of $m$ in each $B_s$. Fix $B_s$ and a parametrization \eqref{lambdaGH}. Then by  \eqref{positive}  we have
\begin{equation}
\nabla V^s_1=\nabla V^s_2=\nabla W^s=0\,.\label{VWzero}
\end{equation}  To show this, one should expand the derivatives with respect to $r$ and $x$ and use an argument similar to that given in \cite{gibbons2006positive,alaee2014thesis}. The details are straightforward but tedious. Since $W^s=0$ on $\mathcal{I}^{\pm}$, we have $W^s\equiv 0$. Also by equations \eqref{lambdaGH} and \eqref{VWzero}, we have $\nabla v=0$ and by Definition \ref{GBdata}, $v$ vanishes at infinity. This implies $v\equiv 0$. Note that in particular this implies there could not be another asymptotic end as $r\to 0$, since $v \propto -\log r$ in that case. Moreover, by definition of $v$ in the parametrization \eqref{lambdaGH} and $v=0$, we have $V^s_1=-V^s_2=$constant. This means for each $B_s$ we have 
\begin{eqnarray}
\begin{aligned}\label{metB_s}
\lambda'_{kk}&=\frac{r^2(1-x)}{2}e^{2V^s_1}\qquad
\lambda'_{jj}=\frac{r^2(1+x)}{2}e^{-2V^s_1},\qquad
\lambda'_{12}=0\qquad v=0\,.
\end{aligned} 
\end{eqnarray}
where $k\neq j$ and $k,j=1,2$. If we consider the last annulus $B_{n'}$ which extends to spatial infinity, i.e. $\mathcal{I}_F$, then by the asymptotic conditions of $\lambda'_{ij}$ in Definition \ref{GBdata} and  $\nabla V^{n'}_1=0$, we obtain $V^{n'}_1=V^{n'}_2\equiv 0$. Moreover, if we consider the common boundary of $B_{n'-1}$ and $B_{n'}$, by the continuity of $V^s_1$ through boundary of $B_s$ and \eqref{changeofcoordinate}, we have
\begin{equation}
4V^{n'-1}_1=\pm\log\left(\frac{\alpha_{(n'-1)1}^k\sigma_{kl}\alpha_{(n'-1)1}^{tl}}{\alpha_{(n'-1)2}^k\sigma_{kl}\alpha_{(n'-1)2}^{tl}}\right)+\log\left(\frac{1+x}{1-x}\right),\quad 0=\alpha_{(n'-1)1}^k\sigma_{kl}\alpha_{(n'-1)2}^{tl}
\end{equation}
where for fixed $k$, $\alpha_{(n'-1)k}^l=(\alpha_{(n'-1)k}^1,\alpha_{(n'-1)k}^2)$ and $\alpha_{(n'-1)k}^{tl}=(\alpha_{(n'-1)k}^1,\alpha_{(n'-1)k}^2)^t$.  These conditions arise by expressing $\lambda'_{ij}$ in $B_{n'-1}$ \eqref{metB_s} in the fixed basis $\xi_{(i)}$ using the transformation \eqref{changeofcoordinate}.  Since  $V_1^{n'-1}=$constant in the above equation and right hand side is a function of $x$ for some $\alpha^{l}_{(n'-1)k}$, then we reach to a contradiction and this implies $n'=1$. This is equivalent to $\Sigma$ having the trivial orbit space, i.e. $\mathcal{B}_{\Sigma}=\mathcal{B}_{\mathbb{R}^4}$. Moreover, we obtain $\lambda'_{ij}=\sigma_{ij}=\frac{r^2}{2}\text{diag}(1+x,1-x)$ and by straightforward computation it implies 
\begin{equation}
-\frac{\det\nabla\lambda'}{2\rho^2}=0,
\end{equation}
Then, the equation \eqref{scalarcurv} reduces to 
\begin{equation}
\Delta_2V=0,\qquad \text{$V$ vanishes on axis and infinity}\,.
\end{equation}
By maximum principle on open set $O_{R,\epsilon}=\{(\rho,z):\quad \epsilon<\rho<R\}$, we have $V\equiv 0$ as $R\to\infty$ and $\epsilon\to 0$. By  \eqref{Airhoz} the one form $\beta^i=A^i_{\rho}\td\rho+A^i_z\td z$ is close and simply connectedness of $\Sigma$ implies that there exists a function $\psi^i$ such that $\beta^i=\td\psi^i$, i.e. $\beta^i$ is exact. Then the metric has the following global representation
\begin{equation}
{h}=\frac{\td\rho^2+\td z^2}{2\sqrt{\rho^2+z ^2}}+\sigma_{ij}\td\left(\phi^i+\psi^i\right)\td\left(\phi^j+\psi^j\right)=\frac{\td\rho^2+\td z^2}{2\sqrt{\rho^2+z ^2}}+\sigma_{ij}\td\gamma^i\td\gamma^j\,.
\end{equation} 
where $\gamma^i$ are new rotational angles with period $2\pi$. Hence, ${h}$ is flat metric and $\Sigma=\mathbb{R}^4$.
\end{proof} 

It is natural to expect this positivity result should extend to GB data that do not belong to $\Xi$.  We will return to this point in the final section. 

\section{Mass-angular momenta inequality}
In \cite{alaee2015proof} a local version of a mass-angular momenta inequality for a class of asymptotically flat, maximal, $U(1)^2$-invariant, vacuum black holes was shown.  The $U(1)^2$ isometry group was assumed to act orthogonally transitively (i.e. there exist two-dimensional surfaces orthogonal to the surfaces of transitivity at every point). There is a question regarding the extension of our proof to the non-vacuum case and considering the general $U(1)^2$-invariant metric equation \eqref{generalmetric}. The main problem in the non-vacuum case is whether angular momenta are conserved quantities and twist potentials exist globally. The ADM angular momenta related to the Killing vector $\xi_{(i)}$ for the GB data set $(\Sigma,h,K,\mu,j)$ is
\begin{equation}
J_{(i)}=\frac{1}{8\pi}\lim_{r\to\infty}\int_{\Sigma}K_{\alpha\beta}n^{\alpha}\xi_{(i)}^{\beta}\td s_h
\end{equation}
This is a well-defined quantity and it is a conserved quantity in $U(1)^2$-invariant vacuum spacetimes. With matter source we show under appropriate conditions it remains a conserved quantity. In the previous section we showed that the ADM mass has lower bound, the right hand of equation \eqref{positivepart}. By the Hamiltonian constraint equation  we have
\begin{equation}\label{Ham}
R_h=|K|^2_h+16\pi\mu\geq |K|^2_h
\end{equation}
if $\mu\geq 0$. In order to prove a local mass angular mometa inequality following the argument of\cite{alaee2015proof} we need to first show the global existence of the potentials
\begin{equation}\label{defY}
\td Y^{(i)}=2\star \left(S^{(i)}\wedge \xi_{(1)}\wedge \xi_{(2)}\right)\qquad S^{(i)}_{\alpha}=K_{\alpha\beta}\xi_{(i)}^{\beta}
\end{equation}
where $\star$ is the Hodge star operator with respect to $h$.
\begin{lemma}\label{existenceofangular}
Consider the GB initial data set $(\Sigma,h,K,\mu,j)$. If $\iota_{\xi_{(i)}}j=0$, then $J_{(i)}$ are conserved and global twist potentials $Y^i$ exist.
\end{lemma}
\begin{proof}
Let $\mathcal{N} \subset \Sigma$  and $\mathcal{S}_1$, $\mathcal{S}_2$ are two 3 dimensional surfaces with isometry subgroup $U(1)^2$ such that $\partial \mathcal{N}=\mathcal{S}_1\cup \mathcal{S}_2$. Then if we consider $\iota_{\xi_{(i)}}j=0$ ($\iota_{\xi_{(i)}}$ is the interior derivative) we have
\begin{eqnarray}
0=\int_{\mathcal{N}}\iota_{\xi_{(i)}}j\, \td\Sigma=\frac{1}{8\pi}\int_{\mathcal{N}}\iota_{\xi_{(i)}}\text{div}K\, \td\Sigma=\frac{1}{8\pi}\int_{\mathcal{S}_1\cup \mathcal{S}_2}K_{\alpha\beta}{n}^{\alpha}{\xi_{(i)}}^{\beta}\, \td s_h=J_{(i)}(\mathcal{S}_2)-J_{(i)}(\mathcal{S}_1)\,.\nonumber
\end{eqnarray}
Thus the angular momenta are conserved quantities. For the second part, let
\begin{equation}
\mathcal{K}^{(i)}=\star \left(S^{(i)}\wedge \xi_{(1)}\wedge \xi_{(2)}\right)\qquad S^{(i)}_{\alpha}=K_{\alpha\beta}\xi_{(i)}^{\beta}\,.
\end{equation}
Then 
\begin{equation}
\td \mathcal{K}^{(i)}=-\iota_{\xi_{(1)}}\iota_{\xi_{(2)}}\td\star S^{(i)}\,.
\end{equation}
then by the Killing property of ${\xi_{(i)}}$ and constraint equation we have  $\star \td\star S^{(i)}=-\iota_{\xi_{(i)}} \text{div} K =-\iota_{\xi_{(i)}}j=0$.  Therefore, since $\Sigma$ is simply connected the potentials $Y^{(i)}$ globally exist. Note that the above result can be extended to $D$-dimensional initial data with $U(1)^{D-2}$ commuting Killing vectors \cite{alaee2014thesis}.
\end{proof} \noindent 

Recall that $t-\phi^i$ symmetric data consists of the subclass of GB initial data with the property that $h_{\alpha \beta} \to h_{\alpha \beta}$ and $K_{\alpha \beta} \to -K_{\alpha \beta}$ under the diffeomorphism $\phi^i \to -\phi^i$\cite{Figueras2011}. It can be shown that for vacuum ($\mu = j = 0$) $t-\phi^i$-symmetric data,  the metric takes the form \eqref{generalmetric} with $A_a^i=0$ and the extrinsic curvature is determined fully from the twist potentials $Y^i$\cite{alaee2014mass}.  Thus this data is characterized by five scalar functions, or equivalently,  the triple $u=(v,\lambda',Y)$, where $v$ is a function, $\lambda'$ is a positive definite symmetric $2\times 2$ matrix, and $Y$ is a column vector\cite{alaee2014mass}. Explicitly, for vacuum $t-\phi^i$ symmetric data, we can express the extrinsic curvature as
\begin{equation} 
K_{\alpha\beta} = 2e^{-2v} S^t_{(\beta} \lambda'^{-1} \Phi_{\alpha)}
\end{equation} where $\Phi^{\alpha} = (\xi^{\alpha} _{(1)},\xi^{\alpha} _{(2)})^t$ is a column vector and $S= (S^1,S^2)^t$ is a column vector with components $S^i$ defined by \eqref{defY} \cite{alaee2014thesis}.  This motivates the following definition.
\begin{Def} Let $(\Sigma,h,K,\mu,j)$ be a GB initial data set with $\mu \geq 0$ and $\iota_{\xi_{(i)}} j=0$.  We define the associated \emph{reduced data} to be the vacuum $t-\phi^i$-symmetric data characterized by the triple $u = (v,\lambda',Y)$ where $(v,\lambda')$ is extracted from the original data and $Y$ is defined in \eqref{defY}.\label{rdata}
\end{Def}
The ADM mass of a given GB data set is bounded below by the ADM mass of its associated reduced data. This can be shown as follows. 
Let introduce the co-frame of one forms $\{\theta^{\alpha}\}$ 
\begin{equation}
\theta^a= e^{v+U}\td x^a,\qquad \theta^{i+2}=e^v\left( \td\phi^i+A^i_a\td x^a\right),
\end{equation} so that the metric can be expressed as
\begin{equation}
 h=(\delta_2)_{ab}\theta^a\theta^b+\lambda'_{ij}\theta^{i+2}\theta^{j+2}\,.
\end{equation} with associated dual frame of basis vectors
\begin{equation}
e_a = e^{-(v+U)}\left(\partial_a - A^i_a \partial_{\phi^i}\right) \qquad e_{i+2} = e^{-v} \partial_{\phi^i}\,.
\end{equation} where $x^a=(\rho,z)$.Then we have
\begin{eqnarray}
\frac{\td Y^{(i)}}{2}&=&\epsilon_{\alpha\beta\gamma\lambda}K^{\beta}_{\delta}\xi_{(i)}^{\delta}\xi_{(1)}^{\gamma}\xi_{(2)}^{\lambda}\td x^{\alpha}\nonumber\\
&=&\epsilon(\partial_a,\partial_b,\partial_{\phi_1},\partial_{\phi_2})K(\td x^b,\partial_{\phi_i})\td x^a\nonumber\\
&=& e^{3v}\epsilon(e_a,e_b,e_1,e_2)K(\theta^b,\partial_{\phi^i}) \theta^a \nonumber \\
&=& e^{3v} \rho\epsilon_{ab}K(\theta^b,e_i) \theta^a\,.
\end{eqnarray} where  $\epsilon_{ab}$ is the volume form on the flat two-dimensional metric.  Noting $K_{bi} = K(e_b,e_i) = K(\theta^b,e_i)$ we read off
\begin{equation}
K_{2i} = \frac{e^{-(4v+U)}}{2\rho} \partial_\rho Y^{(i)} \;,\qquad  K_{1i} = -\frac{e^{-(4v+U)}}{2\rho} \partial_z Y^{(i)} \,.
\end{equation}  Noting that in this basis,
\begin{eqnarray}
|K|^2_h &=& K_{11}^2 + K_{22}^2 + 2K_{12}^2 + 2\lambda'^{ij}K_{1i}K_{1j} + 2\lambda'^{1j}K_{2i}K_{2j} + \lambda'^{ij}\lambda'^{kl}K_{ik}K_{jl}\label{Kinequality} \\
&\geq & 2\lambda'^{ij}K_{1i}K_{1j} + 2\lambda'^{ij}K_{2i}K_{2j} = \frac{e^{-2(4v+U)}}{2\rho^2} \left[\nabla Y^t\lambda'^{-1}\nabla Y\right]\,.\nonumber
\end{eqnarray} where $Y=(Y^{(1)},Y^{(2)})^t$. 
Using \eqref{Ham} and \eqref{positivepart} we arrive at
\begin{equation} 
m\geq \frac{\pi}{4}\int_{\mathcal{B}}\left[2\Delta_2U-\frac{\det\nabla\lambda'}{2\rho^2}+6\abs{\nabla v}^2 + \frac{e^{-6v}}{2\rho^2} \nabla Y^t\lambda'^{-1}\nabla Y \right]\, \td\mu\,.\label{massinequality}
\end{equation} 
Then it follows directly form the results of \cite{alaee2014mass} that we can rewrite the right hand side of equation \eqref{massinequality} as \footnote{There is a sign mistake in \cite{alaee2014mass} because of orientation. The sign of summentaon over rods should be positive.}
\begin{equation}
 \mathcal{M}\equiv\frac{\pi}{4}\int_{\mathcal{B}}\left(-\frac{\det\nabla\lambda'}{2\rho^2}+e^{-6v}\frac{\nabla Y^t\lambda'^{-1}\nabla Y}{2\rho^2}+6\abs{\nabla v}^2\right)\, \, \td\mu+\frac{\pi}{4}\sum_{\text{rods}}\int_{I_s}\log V_s\,\td z\,.\label{Massfunctional}
\end{equation}
which defines the mass functional $\mathcal{M}=\mathcal{M}\left(v,\lambda',Y\right)$ and $V_s$ is defined in Definition \ref{GBdata}-\ref{Vaxis}.  $\mathcal{M}$ evaluates to the ADM mass for vacuum,  $t-\phi^i$ symmetric data. Thus we have shown that $m \geq \mathcal{M}= m_{R}$ where $m_R$ is the ADM mass of the associated reduced data.

One would expect the mass functional is positive definite for all orbit spaces on asymptotically flat $\Sigma$ with positive scalar curvature.  However, positivity of $\mathcal{M}$ has been only established for $\mathcal{B}\in \Xi$ \cite{alaee2014mass}. Thus we have the following conjecture.
\begin{con}
Consider GB initial data set then $\mathcal{M}\left(v,\lambda',Y\right)$ is a non-negative functional for any orbit space.
\end{con}

We set $\bar{u}=(\bar{v},\bar{\lambda}', \bar{Y})$  where $\bar{\lambda'}$ is a symmetric $2 \times 2$ matrix such that $\det\bar{\lambda'}=0$. Consider $\bar{u}$ as a perturbation about some fixed initial data $u_0$ defined in Definition \ref{Def1} . This should consist of five free degrees of freedom, and the apparent restriction $\det\bar{\lambda'} =0$ is simply a gauge choice that preserves the condition $\det \lambda' = \rho^2$ under the perturbation.  Let $\rho_0>0$ and  $\Omega_{\rho_0}\equiv\{(\rho,z,\varphi)| \rho>\rho_0\}$ and select the perturbation $\bar{Y}$ and $\bar{\lambda}$ in $C^\infty_{c}(\Omega_{\rho_0})$. Now for a (unbounded) domain $\Omega$, we introduce the following weighted spaces of $C^1$ functions with norm
\begin{equation}
\norm{f}_{C^1_{s}(\Omega)}=\sup_{x\in \Omega}\{\sigma^{-s}\abs{f}+\sigma^{-s+1}\abs{\nabla f}\}
\end{equation} is finite
with $s<-1$ and $\sigma=\sqrt{r^2+1}$ and for a column vector and a matrix we define respectively
\begin{equation}
\abs{\bar{Y}} \equiv \left(\bar{Y}^t \lambda'^{-1}_0 \bar{Y}\right)^{1/2}\;, \quad
\abs{\bar{\lambda}'} \equiv \left(\text{Tr}\left[\bar{\lambda}'^t\bar{\lambda}'\right]\right)^{1/2}
\end{equation}
Then we define the Banach space $B$ by
\begin{equation}
\norm{\bar{u}}_{B}\equiv\norm{\bar{v}}_{C^1_{s}(\mathbb{R}^3)}+\norm{\bar{\lambda}'}_{C^1_{s}(\Omega_{\rho_0})}+\norm{\bar{Y}}_{C^1_{s}(\Omega_{\rho_0})}
\end{equation}
and similar to \cite{alaee2015proof} we define the extreme class of initial data
\begin{Def}\label{Def1}
The set of \emph{extreme class} $E$ is the collection of data arising from extreme, asymptotically flat, $\mathbb{R}\times U(1)^2$ invariant black holes which consist of triples $u_0 = (v_0,\lambda'_0,Y_0)$ where $v_0$ is a scalar, $\lambda'_0=[\lambda_{ij}]$ is a positive definite $2\times 2$ symmetric matrix, and $Y_0$ is a column vector with the following bounds for $\rho\leq r^2$
\begin{enumerate}
\item $\frac{\nabla Y_0^t\lambda^{-1}_0\nabla Y_0}{X_0}\leq Cr^{-4}$ and $e^{-2v_0}\frac{\nabla Y_0^t\lambda^{-1}_0\nabla Y_0}{X_0}\leq Cr^{-2}$ in $\mathbb{R}^3$
 where $\lambda_0=e^{2v_0}\lambda'_0$
 \item $C_1\rho I_{2\times 2}\leq\lambda_{0}\leq C_2 \rho I_{2\times 2}$ and $C_3\rho^{-1}I_{2\times 2}\leq\lambda^{-1}_{0}\leq C_4\rho^{-1}I_{2\times 2}$ in $\Omega_{\rho_0}$
\item $\rho^2\leq X_0$ in $\mathbb{R}^3$ where $X_0=\det\lambda_0$ and $X_0^2\leq C' \rho^4$ in $\Omega_{\rho_0}$ where $\lim_{\rho_0\to 0}C'=\infty$
\item $\abs{\nabla v_0}^2\leq C r^{-4}$, $\abs{\nabla\ln X_0}^2\leq C\rho^{-2}$ in $\mathbb{R}^3$ and $\abs{\nabla\lambda_0\lambda^{-1}_0}^2\leq C\rho^{-2}$ in $\Omega_{\rho_0}$ 
\item $V=\bar{V}(x)r^{-2}+o_1(r^{-2})$ and $\int_{-1}^{1}\bar{V}(x)\td x=0$ as $r\to\infty$.
\end{enumerate} 
\end{Def} This definition was motivated by studying the geometry of the initial data for the extreme Myers-Perry and black ring solutions.   In has been established that such geometries are local minimizers of the mass amongst suitably nearby data with the same orbit space \cite{alaee2015proof}.  We can now state our second result:
\begin{thm}
Let $(\Sigma,h,K,\mu,j)$ be a GB initial data set with mass $m$ and fixed angular momenta $J_{(1)}$ and $J_{(2)}$ and fixed orbit space $\mathcal{B} \in \Xi$ satisfying $\mu\geq 0$ and $\iota_{\xi_{(i)}}j=0$. Let $u = (v,\lambda',Y)$ describe the associated reduced data as in Definition \ref{rdata} and write $u = u_0 +\bar{u}$ where $u_0$ is extreme data with the same angular momenta and orbit space of the GB initial data set. If $\bar{u}\in B$ is sufficiently small then 
\begin{equation}
 m\geq f(J_{(1)},J_{(2)}) = \mathcal{M}(u_0)
\end{equation} for some $f$ which depends on the orbit space $\mathcal{B}$. Moreover, $m=f(J_{(1)},J_{(2)})$ for GB initial data set in a neighbourhood if and only if the data are extreme data and $\mu=j=0$.
\end{thm}
\begin{proof}
First, consider the GB data with $\mu\geq 0$ and $\iota_{\xi_{(i)}}j=0$. Then by Lemma \ref{existenceofangular}, there exist global potentials $Y^i$ such that $\abs{K}_{\bm{h}}$ satisfies in inequality \eqref{Kinequality} and it yields $m\geq \mathcal{M}(u)$, where $u$ is the associated reduced data. Second, since $u=u_0+\bar{u}$, then all the assumptions of Theorem 1.1 of \cite{alaee2015proof} hold and it follows that there exists $\epsilon>0$ such that if $\norm{\bar{u}}_B < \epsilon$, then $m\geq \mathcal{M}(u_0)$. Finally, by \cite{alaee2015proof} it follows the inequality is saturated if and only if the data is extreme data.
\end{proof}
\subsection*{Acknowledgments}
AA is partially supported by a graduate scholarship from Memorial University. HKK is supported by an NSERC Discovery Grant.  We thank the organizers of the workshop `Constraint equations and mass-momenta inequalities' held May 11-15 2015 at the Fields Institute, Toronto where the work presented here was partially completed.  We would also like to thank Sergio Dain and Marcus Khuri for helpful discussions. 

\appendix
\numberwithin{equation}{section}

\bibliographystyle{unsrt}
\bibliographystyle{abbrv}  
\bibliography{masterfile}

\end{document}